\documentclass{article}

\usepackage{arxiv}
\usepackage[utf8]{inputenc} 
\usepackage[T1]{fontenc}    
\usepackage[numbers]{natbib}
\usepackage{hyperref}       
\usepackage{url}            
\usepackage{booktabs}       
\usepackage{amsfonts}       
\usepackage{nicefrac}       
\usepackage{microtype}      
\usepackage{graphicx}
\usepackage{natbib}
\usepackage{doi}
\usepackage{amsmath}
\usepackage{amssymb}
\usepackage{amsthm}

\newtheorem{theorem}{Theorem}[section]
\newtheorem{lemma}[theorem]{Lemma}

\title{Density
Estimation on the Binary Hypercube using Transformed Fourier-Walsh Diagonalizations}

\author{ \href{https://orcid.org/0000-0003-0178-7709}{\includegraphics[scale=0.06]{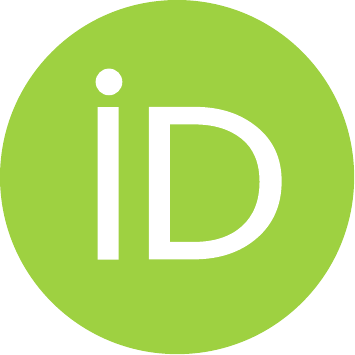}\hspace{1mm}Arthur C. Campello} \\
	Department of Applied Physics\\
	Stanford University\\
	Stanford, CA 94305\\
\texttt{arthur.campello@stanford.edu} \\
}

\renewcommand{\shorttitle}{Binary Density Estimation using Transformed Fourier-Walsh Diagonalizations}

\hypersetup{
pdftitle={Binary Density Estimation using Transformed Fourier-Walsh Diagonalizations},
pdfsubject={statistical science},
pdfauthor={Arthur C.~Campello},
pdfkeywords={density estimation, binary hypercube, Walsh basis, Aitchison and Aitken kernel, diagonalization, positive-definite kernel},
}

\begin{document}
\maketitle

\begin{abstract}
This article focuses on estimating distribution elements over a high-dimensional binary hypercube from multivariate binary data. A popular approach to this problem, optimizing Walsh basis coefficients, is made more interpretable by an alternative representation as a "Fourier-Walsh" diagonalization. Allowing monotonic transformations of the resulting matrix elements yields a versatile binary density estimator: the main contribution of this article. It is shown that the Aitchison and Aitken kernel emerges from a constrained exponential form of this estimator, and that relaxing these constraints yields a flexible variable-weighted version of the kernel that retains positive-definiteness. Estimators within this unifying framework mix together well and span over extremes of the speed-flexibility trade-off, allowing them to serve a wide range of statistical inference and learning problems.
\end{abstract}

\keywords{density estimation \and binary hypercube \and Walsh basis \and Aitchison and Aitken kernel \and diagonalization \and positive-definite kernel}

\section{Introduction}
Though less intuitive than continuous data, high-dimensional binary data appear ubiquitously in modern statistical learning and artificial intelligence. In medicine, critical pieces of information ranging from ocular features \cite{AA1976} to drug trial data to gene expression \cite{SZ2002} take binary forms. Other common binary data types include connectivity and node activation in social and epidemiological networks, survey data, and multivariate binary time series \cite{JR2019}. In machine learning and artificial intelligence (AI) applications, binary data often contain important feature information gleaned from more complex data sets. Binary features have long been used in word image retrieval \cite{ZSH2003} and are used in learning models for facial recognition \cite{LLZZ2015} and fall detection \cite{SJ2018}. Progress in drawing insights from binary data advances both areas dealing with direct and intermediate-form binary information.

This article addresses the estimation of multivariate binary densities from independent observations of $n$-dimensional binary variables whose support is the $\{-1,1\}^n$ hypercube. These densities elucidate complex dependencies across variables and inform conditional probabilities that directly serve statistical learning applications. While estimating all $2^n$ probabilities over the hypercube becomes prohibitively expensive for large $n$, small subsets can still provide strong insights about a dataset. From only two density elements, for example, one can find the expected value of a binary response variable conditional on a specific input from multiple regressor variables. 

In applied settings, the effectiveness of a density estimation scheme can be assessed by three metrics: speed, interpretability, and flexibility. Although these qualities often trade off, practical estimators should run in reasonable time, make inferences in understandable ways, and adapt well to datasets of varying sizes, dimensions, and sparsities. An early triumph in binary density estimation with these attributes came from Aitchison and Aitken's (AA) kernel \cite{AA1976}. The AA kernel measures the proximity between two binary vectors by scaling and exponentiating the number of agreeing indexes between them. For dimension $n$ and number of observations $N$ The approach needs only $O(nN)$ time to estimate a density component. Furthermore, the AA kernel function is shown to be a positive definite kernel over binary spaces, which means it is a reproducing kernel Hilbert space (RKHS) method \cite{M2013} and therefore can leverage a representer theorem. The estimator has been shown to work well for sparse data \cite{GH1993}, but suffers from a flexibility limitation because it has one smoothing parameter for all hypercube dimensions. This rigidity compromises the estimator in cases where densities depend more on agreements along some dimensions than others. 

An alternative approach to binary density estimation uses a weighted sum of orthogonal functions, typically Walsh functions, to find density components \cite{LK1985,CLK1989}. This approach estimates coefficients of a density function's Fourier-Walsh expansion. Since $2^n$ functions exist for $n$ dimensions, optimizing every Walsh coefficient becomes impossible when $n$ is large. At the expense of the method's high flexibility, one may estimate coefficients sparingly or in groups. A notable example of the latter involves using recursive block thresholding to find coefficients in probabilistic polynomial time under certain sparsity conditions \cite{RSLW2013}. Even at extremes of this trade-off, the Fourier-Walsh approach requires more computation than the AA kernel estimator, but proves more versatile due to its higher parameter count. 

While the Fourier-Walsh and AA kernel estimators appear fundamentally distinct, this article shows that a simple transformation of a restricted version of the former yields the latter. It generalizes such transformations with a guarantee of normalization, yielding a powerful binary density estimator whose parameterizations place it on various points of the speed-flexibility trade-off. The resulting estimator uses an interpretable kernel that measures similarity between two $n$-length binary vectors by a signed and weighted sum of $2^n$ variable products acted upon by a monotonic function.

Ways to enhance practical usage of Fourier-Walsh and AA kernel estimators also become apparent. By transparently matching Fourier-Walsh coefficients to corresponding variable products, the construction enables an interpretability-first approach to prioritizing Walsh coefficient optimizations. For the AA kernel, the general estimator's form elucidates an extension of the method to a more flexible dimension-weighted form without compromising its normalization or positive definiteness.

The article is structured as follows: It first demonstrates an intuitive derivation of naturally-ordered Walsh matrices based on how they translate probabilities over the $\{-1,1\}^n$ hypercube -- mapped to $2^n$-length probability vector -- onto expectation values of binary variables and products among them. It justifies using shrinkage coefficients when estimating these expectations from data and shows the proportionality of these to Walsh coefficients. This yields a matrix diagonalization formulation of the Fourier-Walsh estimator with interpretable Walsh coefficients as eigenvalues. From here the article presents an element-wise monotonic transformation of this diagonalization, using special Walsh matrix properties to guarantee the resulting estimator's normalization. This generalized form allows the estimator to incorporate the wide range of activation functions common in machine learning, by themselves or as mixtures. It is then shown that the exponential transformation case with restricted Walsh coefficients yields an AA kernel matrix. Relaxing constraints on the pre-transformation Walsh coefficients introduces a variable-weighted extension of the AA kernel that retains non-negativity and positive definiteness. Following this, the article compares the times required to evaluate to different leave-one-out cross-validation risk functions across variants of the general estimator. It concludes with a discussion of regimes where variants of the general estimator apply and future work invited by the presented estimation approach.

\section{Estimation Theory}

Binary density estimators take in, as inputs, observed binary data points in $n$ dimensions assumed to be i.i.d samples of a binary random variable $X=(X_1,\dots,X_n)$, $X_i\in\{-1,1\}$. The data are used to estimate elements of the distribution of $X$, consisting of $2^n$ nonnegative probability values that sum to one. Even without calculating estimates for all density elements, an effective estimator should guarantee nonnegativity and normalization of its complete output.

The naive approach to this problem simply estimates probabilities by their relative frequencies in the data. These estimates often severely overfit data and especially prove to be ineffective in high dimensions. Practical estimators instead use data to make extrapolative inferences on density elements beyond those corresponding to observations. To ensure reasonable density extrapolations, one can impose typical constraints of neutrality and symmetry. Neutrality means that the effects of new data on the density estimates are independent of existing data. Given two datasets $D$ and $D'$, for example, an estimator $\hat{f}^D_X$ with the property 
\begin{equation}
\hat{f}^{D\cup D'}_X=\frac{1}{|D|+|D'|}\left(|D|\hat{f}^D_X+|D'|\hat{f}^{D'}_X\right)  
\end{equation}
adheres to neutrality. Noting that dataset $D$ could contain a single observation, it becomes clear that this property restricts the estimator to a discrete kernel form. The symmetry constraint means that, given two distinct points in the hypercube $\{-1,1\}^n$, an observation at one point should affect the density estimate at the other in the same way as in the case with the points swapped. All estimators this article presents adhere to both constraints. 

\subsection{Fourier-Walsh Estimation in Matrix Form}

More than a hundred years ago, Joseph L. Walsh cleverly devised a complete set of orthogonal functions $\phi_k$ for $k\in\mathbb{N}$ that yield basis functions in discrete spaces of size $2^n,\:n\in\mathbb{N}$ \cite{W1923}. Specifically, Walsh functions allow one to represent a binary density $f$ defined for $x\in\{0,1\}^n$ as
\begin{equation}
f(x)=\sum_{k\in\{0,1\}^n}c_k\phi_k(x),\qquad\phi_k(x)=(-1)^{\sum_ix_ik_i}.  
\end{equation}
Here, the coefficients $c_k$ act similarly to the coefficients in a Fourier series in that each encodes information pertaining to the entire distribution rather than a single element. This means that estimating only a subset of the involved $2^n$ Walsh coefficients can yield meaningful densities over the entire binary hypercube. For this reason, early and recent research on binary densities involves estimation methods using the Walsh basis \cite{RSLW2013,LK1985,CLK1989}. This article refers to estimators in this class as "Fourier-Walsh" estimators.

Although Walsh coefficients often appear as abstract parameters of equal intrinsic importance in literature \cite{RSLW2013} they carry interpretable information when the components of $X$ are themselves meaningful. Through an intuitive re-derivation of the Walsh decomposition, it is shown that each corresponds to a product of elements in a unique subset of $\{X_1,\dots,X_n\}$ when using the binary support $X_i\in\{-1,1\}$. Thus, some Walsh coefficients carry meaning closely corresponding to input features, while others carry more contrived information encoding products of possibly many features.

Suppose a random vector $\mathbf{r}^{(n)}_X\in \{-1,1\}^{2^n}$ containing the products of all $2^n$ subsets (including the empty set) of $\{X_1,\dots,X_n\}$, generated recursively as
\begin{equation}
\mathbf{r}^{(0)}_X=[1],\qquad \mathbf{r}^{(n+1)}_X=\begin{bmatrix} \mathbf{r}^{(n)}_X \\ X_{n+1}\mathbf{r}^{(n)}_X\end{bmatrix}.
\end{equation}
From this construction, one can also recursively generate a matrix $W^{(n)}\in\{-1,1\}^{2^n\times 2^n}$ whose columns contain the support of $\mathbf{r}^{(n)}_X$. This takes the form 
\begin{equation}
W^{(0)}=[1],\qquad W^{(n+1)}=\begin{bmatrix} W^{(n)} & W^{(n)} \\ W^{(n)} & -W^{(n)}\end{bmatrix}\implies W^{(n)}=\bigotimes^n_{i=1}\begin{bmatrix} 1 & 1 \\ 1 & -1\end{bmatrix}. 
\end{equation}
Interestingly, this matrix is a naturally ordered (or Haramard) Walsh matrix. It is symmetric and has the important property $W^{(n)}\left[W^{(n)}\right]^\intercal=\left[W^{(n)}\right]^2=2^nI_{2^n}$. It follows from these definitions that $\mathbb{E}\left[\mathbf{r}^{(n)}_X\right]=W^{(n)}\mathbf{p}$, where
\begin{equation}
p_j=\mathbb{P}\left\{\mathbf{r}^{(n)}_X=W^{(n)}_{:,j}\right\}.
\end{equation}
Noting that the construction of $\mathbf{r}^{(n)}_X$ means $\left[\mathbf{r}^{(n)}_X\right]_{2^{k-1}+1}=X_k$, one can also write
\begin{equation}\label{eq:6}
p_j=\mathbb{P}\left\{\bigcap^n_{k=1}\left(X_k=W^{(n)}_{2^{k-1}+1,j}\right)\right\}.
\end{equation}
The vector $\mathbf{p}$ encodes the probability distribution of $X$ over the hypercube as a vector and must satisfy the constraints $\mathbf{1}^\intercal\mathbf{p}=1$ and $\mathbf{p}\succeq\mathbf{0}$. From here forward, $\hat{\mathbf{p}}$ refers to the estimator of $\mathbf{p}$. This vector mapping also defines a normalized "counts" vector $\mathbf{p}_k$ that encodes observed instances of $X$.

Given these definitions, $\left[W^{(n)}\mathbf{p}_k\right]_j$ encodes the sample mean of the product $\left[\mathbf{r}^{(n)}_X\right]_j$ and $\left[W^{(n)}\hat{\mathbf{p}}\right]_j$ the expectation of $\left[\mathbf{r}^{(n)}_X\right]_j$ associated with the estimation of the binary density. For now, let $\left[W^{(n)}\hat{\mathbf{p}}\right]_j$ be equal to $\left[W^{(n)}\mathbf{p}_k\right]_j$ multiplied by a shrinkage factor $b_j\in [0,1]$. This means $W^{(n)}\hat{\mathbf{p}}=\mathrm{diag}(\mathbf{b})W^{(n)}\mathbf{p}_k$, where $\mathbf{b}\in[0,1]^{2^n}$ is now a shrinkage vector and produces the estimator
\begin{equation}\label{eq:7}
\hat{\mathbf{p}}=\frac{1}{2^n}W^{(n)}\mathrm{diag}(\mathbf{b})W^{(n)}\mathbf{p}_k.
\end{equation}
This form requires constraints on $\mathbf{b}$ to ensure $\hat{\mathbf{p}}$ is a true density. Rearranging Equation \ref{eq:7} gives
$ \mathbf{b}=\mathrm{diag}^{-1}(W^{(n)}\mathbf{p}_k)W^{(n)}\hat{\mathbf{p}}.$
Since the first row of this matrix is trivially $\mathbf{1}^\intercal$, the equality means $\mathbf{1}^\intercal\hat{\mathbf{p}}=1\iff b_1=1$. For $\hat{\mathbf{p}}\succeq\mathbf{0}$ to hold, $\mathbf{b}$ must be a convex combination of the columns of $\left[\mathrm{diag}^{-1}(W^{(n)}\mathbf{p}_k)W^{(n)}\right]$. Geometrically, this means $\mathbf{b}_{2:2^n}$ must lie inside the simplex formed by the columns of $\left[\mathrm{diag}^{-1}(W^{(n)}\mathbf{p}_k)W^{(n)}\right]$ with its first row removed.

Equation \ref{eq:7} is equivalent to a Fourier-Walsh estimator with coefficients proportional to elements $b_j\left[W^{(n)}\mathbf{p}_k\right]_j$. Here, the constraint $b_j\in[0,1]$ is justified in addition to the aforementioned constraints on $\mathbf{b}$. The chosen binary basis $X_i\in\{-1,1\}$ means $\left[\mathbf{r}^{(n)}_X\right]_j\in\{-1,1\}$ and $\left[W^{(n)}\mathbf{p}_k\right]_j$ represents its mean from binary samples. Suppose an analogous univariate random variable $Y\in\{-1,1\}$ with $\mathbb{E}[Y]=q$ and $N$ observed outcomes of $Y$ with a sample mean $\bar{y}$. If one estimates $q$ using a factor $b$ as $\hat{q}=b\bar{y}$, then the square error minimizing $b^*$ given a true $q$ is
\begin{equation}
b^*=\underset{b}{\mathrm{argmin}}\:\mathbb{E}_q\left[\left(b\bar{y}-q\right)^2\right]=q\frac{\mathbb{E}\left[\bar{y}\right]}{\mathbb{E}\left[\bar{y}^2\right]}=\frac{q^2}{\mathbb{E}\left[\bar{y}^2\right]}=\frac{Nq^2}{(N-1)q^2+1}.
\end{equation}
This equation constrains $b^*\in[0,1]\:\forall\:q\in[-1,1],N\in\mathbb{N}$. As expected, $b^*(-1)=b^*(1)=1$ and $b^*(0)=0$; note that the optimal $b^*$ varies substantially over $q$ even for large values of $N$. This shrinkage technique works similarly to others in applied statistics, such as the James–Stein estimator \cite{JS1961} and lasso and ridge regressions. Note that the case $\mathbf{b}=\mathbf{1}$ corresponds to no regularization and reduces Equation \ref{eq:7} to $\hat{\mathbf{p}}=\mathbf{p}_k$, that is, the data frequency estimate. The fully regularized case of $\mathbf{b}=[1\:\mathbf{0}^\intercal]^\intercal$ yields the uniform estimate $\hat{\mathbf{p}}=\mathbf{1}/2^n$. A later section discusses using cross-validation to optimize $\mathbf{b}$ within these extremes.

The derivation of Equation \ref{eq:7} elucidates that $b_j$ regularizes the expectation of $\left[\mathbf{r}^{(n)}_X\right]_j$ used in the estimator. This means that $n\choose k$ elements of $\mathbf{b}$ -- and Walsh coefficients -- correspond to products of elements in subsets of $\{X_i,\dots,X_n\}$ of size $k\in\mathbb{N}$. Trivially $b_1$ corresponds to the empty set whose product is 1, further justifying setting $b_1=1$. In applications where variables $X_i$ reflect interpretable information, therefore, elements of $\mathbf{b}$ associated products of small subsets of $\{X_i,\dots,X_n\}$ carry more meaning. This creates an intuition hierarchy of Walsh coefficients that can inform optimization choices when $n$ is large. To specify this hierarchy, define $S^{(n)}_k\subset\mathbb{N}$ to be the set of indexes of $\mathbf{r}^{(n)}_X$ corresponding to products of $k$ variables. From the construction of $\mathbf{r}^{(n)}_X$, it follows that
\begin{equation}
S^{(n)}_0=\{1\},\qquad S^{(n+1)}_k = S^{(n)}_k\cap \left\{x+2^n\big|\:x\in S^{(n)}_{k-1}\right\}.
\end{equation}
Note that $S^{(n)}_1=\{2^{x-1}+1\:|\:x\in[n]\}$, which matches the indexes of $W^{(n)}$ elements in Equation \ref{eq:6}. An intuition-first approach to optimizing $\mathbf{b}$ prescribes prioritizing indexes in sets $S^{(n)}_k$ of low $k$. In a case where one considers any combination of more than three variables in $\{X_i,\dots,X_n\}$ uninterpretable, for example, only $(n^3+5n+6)/6$ out of the $2^n$ elements require optimization for an interpretable estimator. 

\subsection{Monotonic Transformations of Fourier-Walsh Diagonalization Elements}

At its core, the Fourier-Walsh estimator in matrix form defines a similarity metric between two points on the $\{-1,1\}^n$ hypercube. For hypercube points assigned to indexes $i$ and $j$ in $\mathbf{p}$, Equation \ref{eq:7} gives the kernel $K_{ij} = \mathbf{b}^\intercal(W^{(n)}_{:,i}\odot W^{(n)}_{:,j})/2^n$. This intuitive kernel sums elements of $\mathbf{b}$ with a factor $(\pm1)$ on $b_k$ depending on whether the two input points share the same $\left[\mathbf{r}^{(n)}_X\right]_k$; it then divides this by the number of elements. In many cases, monotonic transformations of this kernel, which preserve its interpretable ordering, can enhance it. For example, transforming elements $K_{ij}$ using a nonnegative function guarantees nonnegative density estimates without any restriction on $\mathbf{b}$. Furthermore, some transformations allow the kernel to be positive definite without the requirement $\mathbf{b}\succ\mathbf{0}$ as in Equation \ref{eq:7}.

Estimating binary densities using such transformed Fourier-Walsh matrices requires a guarantee of normalization. Here s fact specific to Fourier-Walsh matrices is proven to facilitate a normalized generalization of Equation \ref{eq:7} with monotonically transformed elements. Since products of Walsh functions are themselves Walsh functions, the columns and rows of Walsh matrices must be closed under element-wise multiplication. Given this fact, suppose a mapping matrix $\mathcal{M}^{(n)}\in\mathbb{N}^{2^n\times 2^n}$ where $W^{(n)}_{:,\mathcal{M}^{(n)}_{ij}} = W^{(n)}_{:,i}\odot W^{(n)}_{:,j}$.

\begin{lemma}\label{lem:1}
$\mathcal{M}^{(n)}$ exists and each of its rows and columns contain unique elements in $[2^n]\subset\mathbb{N}$.
\end{lemma} 

\begin{proof}
In the base case $n=0$, it is evident that $W^{(0)}=[1]\implies\mathcal{M}^{(0)}=[1]$. From the recursive construction of $W^{(n)}$, the following hold true for indexes $i,j\in\{1,\dots,2^n\}$:
\begin{equation}
W^{(n)}_{:,i}\odot W^{(n)}_{:,j}=W^{(n)}_{:,\mathcal{M}^{(n)}_{ij}}\implies
\begin{cases} 
W^{(n+1)}_{:,i}\odot W^{(n+1)}_{:,j}=W^{(n+1)}_{:,2^n+i}\odot W^{(n+1)}_{:,2^n+j}=W^{(n+1)}_{:,\mathcal{M}^{(n)}_{ij}} \\
W^{(n+1)}_{:,i}\odot W^{(n+1)}_{:,2^n+j}=W^{(n+1)}_{:,2^n+i}\odot W^{(n+1)}_{:,j}=W^{(n+1)}_{:,2^n+\mathcal{M}^{(n)}_{ij}}
\end{cases}.
\end{equation}
Thus,
\begin{equation}\label{eq:11}
\mathcal{M}^{(n+1)}=\begin{bmatrix}
\mathcal{M}^{(n)} & 2^nJ_{2^n}+\mathcal{M}^{(n)} \\
2^nJ_{2^n}+\mathcal{M}^{(n)} & \mathcal{M}^{(n)}
\end{bmatrix},
\end{equation}
where $J_m$ is the $m\times m$ all-ones matrix. Note from Equation \ref{eq:11} that if every row and every column of $\mathcal{M}^{(n)}$ contains all integers $[2^n]$, then $\mathcal{M}^{(n+1)}$ will have the same property for integers $[2^{n+1}]$. Because this is true for the base case $\mathcal{M}^{(0)}=[1]$ it holds for all $\mathcal{M}^{(n)}$. This concludes the proof.
\end{proof}

With the mapping matrix $\mathcal{M}^{(n)}$ defined, one can write
\begin{equation}
\left[W^{(n)}\mathrm{diag}(\mathbf{b})W^{(n)}\right]_{ij}=\sum^{2^n}_{k=1}b_kW^{(n)}_{ik}W^{(n)}_{jk}=\mathbf{b}^\intercal\left(W^{(n)}_{:,i}\odot W^{(n)}_{:,j}\right)=\mathbf{b}^\intercal W^{(n)}_{:,\mathcal{M}^{(n)}_{ij}}.
\end{equation}
By Lemma \ref{lem:1}, this means that all rows and columns of $W^{(n)}\mathrm{diag}(\mathbf{b})W^{(n)}$ contain the same elements. This also applies to element-wise transformations of $W^{(n)}\mathrm{diag}(\mathbf{b})W^{(n)}$.
Denoting $(Q)^f$ to mean the element-wise action of $f:\mathbb{R}\to\mathbb{R}$ on matrix $Q$, This fact means
\begin{equation}
\left(W^{(n)}\mathrm{diag}(\mathbf{b})W^{(n)}\right)^f\mathbf{1}=\left[\left(W^{(n)}\mathrm{diag}(\mathbf{b})W^{(n)}\right)^f\mathbf{1}\right]_1\mathbf{1}.
\end{equation}
Since $W_{:,1}=\mathbf{1}$, it follows that
\begin{equation}
\left[\left(W^{(n)}\mathrm{diag}(\mathbf{b})W^{(n)}\right)^f\mathbf{1}\right]_1=\sum^{2^n}_{j=1}f\left(\sum^{2^n}_{k=1}b_kW^{(n)}_{1k}W^{(n)}_{jk}\right)=\left(\mathbf{b}^\intercal W^{(n)}\right)^f\mathbf{1}.
\end{equation}
This normalization factor yields the complete general-form estimator
\begin{equation}\label{eq:15}
\hat{\mathbf{p}}=\frac{\left(W^{(n)}\mathrm{diag}(\mathbf{b})W^{(n)}\right)^f}{\left(\mathbf{b}^\intercal W^{(n)}\right)^f\mathbf{1}}\mathbf{p}_k.
\end{equation}
Note that $f(x)=x$ and $b_1=1$ yield $\left(\mathbf{b}^\intercal W^{(n)}\right)^f\mathbf{1}=\mathbf{b}^\intercal W^{(n)}\mathbf{1}=2^n$, which is consistent with Equation \ref{eq:7}.

Elements of the matrix in the numerator of Equation \ref{eq:15} are evaluated in time proportional to the number of nonzero elements in $\mathbf{b}$ regardless of $f$. However, the normalization factor for nontrivial $\mathbf{b}$ and nonlinear $f$ typically takes $O(n2^n)$ time to calculate using the Fast Walsh Transform algorithm \cite{HB2011}, regardless of the number of nonzero elements in $\mathbf{b}$. In the cases of logistic and exponential transformations, a certain restriction on $\mathbf{b}$ can greatly reduce the normalization times to $O(1)$ and $O(n)$ respectively. These improvements are referenced in Table~\ref{tab:table}.

Define a vector $\mathbf{w}\in[0,1]^n$ and an associated $\mathbf{b}_{\mathbf{w}}$ such that $\mathbf{b}_{\mathbf{w}}=\sum^n_{k=1}w_k\hat{e}_{2^{k-1}+1}$; this means $\mathbf{b}_{\mathbf{w}}$ has nonzero elements only at indexes in $S^{(n)}_1$. Since all column vectors in $\{W^{(n)}_{:,k}|k\in S^{(n)}_1\}$ are anti-symmetric -- in the sense $\mathbf{v}=-\mathrm{flip}(\mathbf{v})$ -- it follows that $W^{(n)}\mathbf{b}_{\mathbf{w}}$ is also anti-symmetric. From here, the propriety of the sigmoid/logistic function $f(x)+f(-x)=1$ allows one to write
\begin{equation}
f=\frac{1}{1+\gamma^{-x}}\implies \mathbf{1}^\intercal\left(W^{(n)}\mathbf{b}_{\mathbf{w}}\right)^f=\left(\mathbf{b}^\intercal_{\mathbf{w}}W^{(n)}\right)^f\mathbf{1}=2^n/2,
\end{equation}
$\forall\:\gamma\in\mathbb{R}_+$. This is evaluated in $O(1)$ time. Turning to the exponential case, one notes that
\begin{equation}
W^{(n)}\mathbf{b}_{\mathbf{w}} =
\begin{bmatrix} 1 \\ 1 \end{bmatrix}\otimes W^{(n-1)}\left[\mathbf{b}_{\mathbf{w}}\right]_{1:2^{n-1}}W^{(n-1)}+w_n\begin{bmatrix} \mathbf{1} \\ -\mathbf{1} \end{bmatrix}.
\end{equation}
This means 
\begin{equation}
f(x) = \gamma^{x}\implies 
\left(W^{(n)}\mathbf{b}_{\mathbf{w}}\right)^f = \begin{bmatrix} \gamma^{w_n} \\ \gamma^{-w_n} \end{bmatrix}\otimes \left(W^{(n-1)}\left[\mathbf{b}_{\mathbf{w}}\right]_{1:2^{n-1}}W^{(n-1)}\right)^f=\bigotimes^{n-1}_{i=0}\begin{bmatrix} \gamma^{w_{n-i}} \\ \gamma^{-w_{n-i}} \end{bmatrix}
\end{equation}
and, finally,
\begin{equation}\label{eq:19}
f(x) = \gamma^{x}\implies=\left(\mathbf{b}^\intercal_{\mathbf{w}}W^{(n)}\right)^f\mathbf{1}=\prod^n_{i=1}\left(\gamma^{w_i}+\gamma^{-w_i}\right).
\end{equation}
Equation \ref{eq:19} evaluates in $O(n)$ time.

The powerful flexibility of kernel transformation enables this binary density estimator to employ the wide range of activation functions used in applied machine learning. These include exponential, logistic/sigmoid, step, ReLU, $\tanh$, ELU functions, and many others. The nonnegative natures of the first four listed functions make them especially useful in guaranteeing nonnegative density estimates. The choices among these functions can depend on the performance of cross-validation, the context of the application, and the desired evaluation speed. 

An additional benefit of the presented matrix formulation is that convex combinations of normalized matrices may also be used, allowing mixtures of transformed kernel estimators. Such a mixed estimator would take the form
\begin{equation}
\hat{\mathbf{p}}=\left[\sum^m_{i=1}c_i\frac{\left(W^{(n)}\mathrm{diag}(\mathbf{b}_i)W^{(n)}\right)^{f_i}}{\left(\mathbf{b}^\intercal_i W^{(n)}\right)^{f_i}\mathbf{1}}\right]\mathbf{p}_k,
\end{equation}
where $\sum^m_{i=1}c_i=1$, $c_i>0$.

\subsection{Aitchison Aitken Kernel from Exponential Fourier-Walsh Matrix}

This section demonstrates that using an exponential function in Equation \ref{eq:15} with $\mathbf{b}=\mathbf{b}_\mathbf{w}$ type restrictions yields the AA kernel estimator; this is despite the seemingly fundamental differences between the AA and Fourier-Walsh approaches to estimation. The AA kernel gives a similarity metric between two $n$-length binary vectors $\mathbf{x}_i$ and $\mathbf{x}_j$ with parameter $\lambda\in [1/2,1]$ as
$$K(\mathbf{x}_i,\mathbf{x}_j;\lambda)=\lambda^{n-d(\mathbf{x}_i,\mathbf{x}_j)}(1-\lambda)^{d(\mathbf{x}_i,\mathbf{x}_j)},$$
where $d(\mathbf{x}_i,\mathbf{x}_j)$ gives the number of elements where $\mathbf{x}_i$ and $\mathbf{x}_j$ differ. This kernel guarantees normalization and nonnegative density estimation \cite{AA1976}. Supposing that $\mathbf{x}_i,\mathbf{x}_j\in \{-1,1\}^n$, one can write
\begin{equation}
\mathbf{x}^\intercal_i\mathbf{x}_j=\left[n-d(\mathbf{x}_i,\mathbf{x}_j)\right]-d(\mathbf{x}_i,\mathbf{x}_j)\implies d(\mathbf{x}_i,\mathbf{x}_j)=\frac{n-\mathbf{x}^\intercal_i\mathbf{x}_j}{2},\quad n-d(\mathbf{x}_i,\mathbf{x}_j)=\frac{n+\mathbf{x}^\intercal_i\mathbf{x}_j}{2}.
\end{equation}
The AA kernel equation now becomes
\begin{equation}\label{eq:21}
K(\mathbf{x}_i,\mathbf{x}_j;\lambda)=\sqrt{\lambda}^{n+\mathbf{x}^\intercal_i\mathbf{x}_j}\sqrt{1-\lambda}^{n-\mathbf{x}^\intercal_i\mathbf{x}_j}=\sqrt{\lambda(1-\lambda)}^n\sqrt{\frac{\lambda}{1-\lambda}}^{\mathbf{x}^\intercal_i\mathbf{x}_j}.
\end{equation}
Using the binary hypercube to $2^n$-vector mapping implicit in Equation \ref{eq:6}, one can write $\mathbf{x}^\intercal_i\mathbf{x}_j=\left[W^{(n)}\mathrm{diag}(\mathbf{b}_{\mathbf{1}})W^{(n)}\right]_{ij}$, where $\mathbf{b}_{\mathbf{1}}$ corresponds to $\mathbf{b}_{\mathbf{w}}$ in the case $\mathbf{w}=\mathbf{1}$. Using Equation \ref{eq:19}, one notes that
\begin{equation}
f(x)=\sqrt{\frac{\lambda}{1-\lambda}}^x\implies \left[\left(\mathbf{b}^\intercal_{\mathbf{1}}W^{(n)}\right)^f\mathbf{1} \right]^{-1}=\left(\sqrt{\frac{\lambda}{1-\lambda}}+\sqrt{\frac{1-\lambda}{\lambda}}\right)^{-n}=\sqrt{\lambda(1-\lambda)}^{n},
\end{equation}
in agreement with the AA kernel normalization in Equation \ref{eq:21}. Therefore, the AA kernel estimator is equivalent to 
\begin{equation}
\hat{\mathbf{p}}_{\mathrm{AAK}}=
\frac{\left(W^{(n)}\mathrm{diag}(\mathbf{b}_{\mathbf{1}})W^{(n)}\right)^f}{\left(\mathbf{b}^\intercal_{\mathbf{1}} W^{(n)}\right)^f\mathbf{1}}\mathbf{p}_k
,\qquad f(x)=\sqrt{\frac{\lambda}{1-\lambda}}^x,
\end{equation}
a restricted form of the estimator in Equation \ref{eq:15}. 

Given this equation, one may relax $\mathbf{b}$ from $\mathbf{b}_{\mathbf{1}}$ to $\mathbf{b}_{\mathbf{w}}$ to yield a variable-weighted version of the AA kernel. This allows the elements of $\mathbf{w}$ to parameterize different "smoothing" levels along the $n$ dimensions of the hypercube. A variable-weighted kernel reflects realistic cases where similarities of two points in some indexes matter more to overall similarity than similarities in other indexes. For convenience, a reparameterization $\gamma = \sqrt{\lambda/(1-\lambda)}$ is introduced. The bounds $\lambda\in[1/2,1]$ correspond to $\gamma \in[1,\infty)$. Letting $\mathbf{b}_{\mathbf{1}}\to\mathbf{b}_{\mathbf{w}}$ yields the weighted AA kernel estimator
\begin{equation}\label{eq:24}
\hat{\mathbf{p}}_{\mathrm{WAAK}}=\frac{\left(W^{(n)}\mathrm{diag}(\mathbf{b}_{\mathbf{w}})W^{(n)}\right)^f}{\left(\mathbf{b}^\intercal_{\mathbf{w}} W^{(n)}\right)^f\mathbf{1}}\mathbf{p}_k,\qquad f(x)=\gamma^x.
\end{equation}

It is now shown that one may write the numerator of Equation \ref{eq:24} as a series of Kronecker products, easing computation and elucidating certain properties about the kernel's positive definiteness. Noting that the product $W^{(n-1)}\mathrm{diag}\left([a\:\mathbf{0}^\intercal]\right)W^{(n-1)}$ gives $aJ_{2^{n-1}}$, the recursive form of $W^{(n)}$ means
\begin{equation}
W^{(n)}\mathrm{diag}(\mathbf{b}_{\mathbf{w}})W^{(n)} = 
\begin{bmatrix} 1 & 1 \\ 1 & 1 \end{bmatrix}\otimes W^{(n-1)}\mathrm{diag}\left([\mathbf{b}_\mathbf{w}]_{1:2^{n-1}}\right)W^{(n-1)}+\begin{bmatrix} 1 & -1 \\ -1 & 1 \end{bmatrix}w_nJ_{2^{n-1}}.
\end{equation}
This means
\begin{equation}
f(x)=\gamma^x\implies\left(W^{(n)}\mathrm{diag}(\mathbf{b}_\mathbf{w})W^{(n)}\right)^f=
\begin{bmatrix}
\gamma^{w_n} & \gamma^{-w_n} \\ \gamma^{-w_n} & \gamma^{w_n}
\end{bmatrix}\otimes\left(W^{(n-1)}\mathrm{diag}([\mathbf{b}_\mathbf{w}]_{1:2^{n-1}})W^{(n-1)}\right)^f.
\end{equation}
Continuing the recursion, the complete weighted AA kernel estimator can be written as 
\begin{equation}\label{eq:27}
\hat{\mathbf{p}}_{\mathrm{WAAK}}=\left[\prod^n_{j=1}\left(\gamma^{w_j}+\gamma^{-w_j}\right)\right]^{-1}\bigotimes^{n-1}_{i=0}\begin{bmatrix}
\gamma^{w_{n-i}} & \gamma^{-w_{n-i}} \\ \gamma^{-w_{n-i}} & \gamma^{w_{n-i}}
\end{bmatrix}\mathbf{p}_k.
\end{equation}
In this form, elements of the weighted AA kernel matrix can be evaluated in $O(n)$ time. Note also that since a $2\times 2$ matrix with on-diagonal elements $\gamma^w$ and off diagonal elements $\gamma^{-w}$ has eigenvalues $\gamma^w-\gamma^{-w}$ and $\gamma^w+\gamma^{-w}$, the matrices in the Kronecker product in Equation \ref{eq:27} are all positive definite when $\mathbf{w}\succ\mathbf{0}$ and $\gamma>1$. Since a Kronecker product of two positive definite matrices is also positive definite, it follows that the weighted AA kernel estimator uses a positive definite kernel function. This analysis rederives the known result that the AA kernel is positive definite \cite{M2013} and extends this property to its variable-weighted generalization.

\section{Cross Validation}

The flexibility of the presented general-form estimator is driven by its variable, and possibly high, number of parameters. These include elements of $\mathbf{b}$, parameters defining $f$, and weights $c_i$ applied to matrices when mixing estimators. Practical optimization of these parameters typically involves a cross-validation scheme. This section outlines the common leave-one-out approach to cross-validation using the squared error (SE) and Kullback-Leibler (KL) loss functions. It specifically focuses on risk function evaluation times across variants of the general-form estimator.

For the optimizations described, suppose a "true" probability distribution $\mathbf{p}$ and the estimator $\hat{\mathbf{p}}_\lambda$ parameterized by $\lambda$. The optimal $\lambda^*$ minimizes the expectation of a loss function $L(\mathbf{p}, \hat{\mathbf{p}}_\lambda)$ i.e., the risk. The SE and KL loss functions are respectively defined as
\begin{equation}
L_{\mathrm{SE}}(\mathbf{p}, \hat{\mathbf{p}}_\lambda)=||\hat{\mathbf{p}}_\lambda-\mathbf{p}||^2_2,\qquad L_{\mathrm{KL}}(\mathbf{p}, \hat{\mathbf{p}}_\lambda)=\sum^{2^n}_{j=1}\mathbf{p}_j\log\left(\mathbf{p}_j\big/\left[\hat{\mathbf{p}}_\lambda\right]_j\right).
\end{equation}
To implement the leave-one-out technique, let $K\in[2^n]$ be a multiset encoding indexes of observations. Also, define $\hat{\mathbf{p}}^{(k)}_\lambda$ as the estimate made without an observed data point corresponding to $k\in K$. By the law of the unconscious statistician,
\begin{equation}
\mathbb{E}\left[\hat{\mathbf{p}}^\intercal_\lambda\mathbf{p}\right]=\frac{1}{|K|}\sum_{k\in K}\mathbb{E}\left[\left(\hat{\mathbf{p}}^{(k)}_\lambda\right)^\intercal\hat{e}_k\right]\implies \arg\min_\lambda \mathbb{E}\left[L_{\mathrm{SE}}(\mathbf{p}, \hat{\mathbf{p}}_\lambda)\right]=\arg\min_\lambda\left[\hat{\mathbf{p}}^\intercal_\lambda\hat{\mathbf{p}}_\lambda-\frac{2}{|K|}\sum_{k\in K}\left(\hat{\mathbf{p}}^{(k)}_\lambda\right)^\intercal\hat{e}_k\right].
\end{equation}
For the KL loss, the cross-validation optimization is given by 
\begin{equation}
\arg\min_\lambda \mathbb{E}\left[L_{\mathrm{KL}}(\mathbf{p}, \hat{\mathbf{p}}_\lambda)\right]=\arg\max_\lambda \sum_{k\in K}\log\left[\left(\hat{\mathbf{p}}^{(k)}_\lambda\right)^\intercal\hat{e}_k\right],
\end{equation}
using a discrete adaptation of the known KL cross validation optimizer for density estimation over continuous spaces \cite{H1987}.

One can now impose the familiar form $\hat{\mathbf{p}}_\lambda=Q_\lambda\mathbf{p}_k=Q_\lambda\left[(1/|K|)\sum_{k\in K}\hat{e}_k\right]$ for some generic symmetric estimator matrix $Q_\lambda\in\mathbb{R}^{2^n\times 2^n}$. This means that evaluating $\left(\hat{\mathbf{p}}^{(k)}_\lambda\right)^\intercal\hat{e}_k$ involves knowing $|K|-1$ elements of $Q_\lambda$. Therefore, the computation of the KL risk function requires the evaluation of $|K|(|K|-1)/2$ elements of $Q_\lambda$. In the SE risk case, the second term also involves knowing $|K|(|K|-1)/2$ elements of $Q_\lambda$, while the first requires evaluating $|K|(|K|+1)/2$ elements of $Q^2_\lambda$. 
Squared-matrix elements are particularly simple to calculate for numerators of Equations \ref{eq:7} and \ref{eq:27} since
\begin{equation}
\left(W^{(n)}\mathrm{diag}(\mathbf{b})W^{(n)}\right)^2=2^n\left(W^{(n)}\mathrm{diag}(\mathbf{b}\odot\mathbf{b})W^{(n)}\right)
\end{equation}
and 
\begin{equation}
f(x)=\gamma^x\implies\left[\left(W^{(n)}\mathrm{diag}(\mathbf{b}_\mathbf{w})W^{(n)}\right)^f\right]^2=
\bigotimes^{n-1}_{i=0}\begin{bmatrix}
\gamma^{2w_{n-i}}+\gamma^{-2w_{n-i}} & 2 \\ 2 & \gamma^{2w_{n-i}}+\gamma^{-2w_{n-i}}
\end{bmatrix},
\end{equation}
by the Kronecker mixed product property.

This section concludes with a tabulation of the computation times required to normalize and evaluate matrix and squared-matrix elements of useful estimators derived from restrictions of Equation \ref{eq:15}. Table~\ref{tab:table} shows these computation times, where $b$ gives the number of nonzero (or non-constant) elements of $\mathbf{b}$ and the fast Walsh transform \cite{HB2011} is used for listed operations that take time $O(n2^n)$.

\begin{table}
\caption{Evaluation times of normalization and finding one element of an estimator matrix and its square as functions of $n$, the hypercube dimension, and $b$, the number of non-constant diagonal elements. These are reported for the general-form estimator in Equation \ref{eq:15} for various monotonic functions and restrictions on $\mathbf{b}$.}
\centering
\begin{tabular}{llll}
    \toprule
    Estimator Restrictions & 
    Normalization time &
    Matrix element time &
    Squared-Matrix element time \\
    \midrule
    $f(x)=x, b_1=1$  & 
    $O(1)$ & 
    $O(b)$ & 
    $O(b)$ \\
    $f(x)=\gamma^x, \mathbf{b}=\mathbf{b}_\mathbf{w}$  &
    $O(n)$ & 
    $O(n)$ & 
    $O(n)$ \\
    $f(x)=\left(1+\gamma^{-x}\right)^{-1}, \mathbf{b}=\mathbf{b}_\mathbf{w}$ &
    $O(1)$ & 
    $O(n)$ & 
    $O(n2^n)$ \\
    No restriction &
    $O(n2^n)$ & 
    $O(b)$ & 
    $O(n2^n)$ \\
    \bottomrule
\end{tabular}
\label{tab:table}
\end{table}

\section{Summary and Discussion}

This article presents a powerful binary density estimator built on element-wise monotonic transformations of Fourier-Walsh diagonalizations. To accomplish this, the article first provides an intuitive rederivation of the Fourier-Walsh decomposition in the form of a diagonalization. In this form, Walsh coefficients are shown to relate to unique products of constituent univariate binary variables. A specific property of Walsh matrices is then shown that enables normalization of an estimator arising from any elementwise transformation function. It is then elucidated how the AA kernel arises from the above process with a generic exponential transformation, and a variable-weighted extension of the kernel is introduced that retains its desirable properties. Finally, the implementations of leave-one-out cross-validation risk functions are outlined for squared error and Kullback-Leibler loss functions and their computation times are compared across estimators. The flexibility, speed, and interpretability of this new estimator under different constraints make it an ideal candidate for use in a wide range of estimation and learning applications.

The comparison made in Table~\ref{tab:table} shows that variants of the proposed estimator under different constraints serve best in different regimes of data science. For problems of up to $n\approx 20$ dimensions -- i.e. binary inputs -- computers can safely handle $O(n2^n)$ time operations and the estimator in Equation \ref{eq:15} may apply in its most general form. Learning in this setting could involve iterating over a large number of transformation functions and exploring mixtures of several different estimation matrices. Problems of approximately 20 dimensions and 40 data points have been cited as typical in applied binary density estimation \cite{GH1993}. At the other extreme, one could consider a high-dimensional case of $n$ up to $n\approx 10^4$. Here, any approach other than the introduced weighted AA kernel and untransformed Fourier-Walsh diagonalization -- with heavily restricted $\mathbf{b}$ -- becomes highly intractable. 

These extremes not only showcase the high versatility of the general form estimator, but also invite the possibility of variable selection when faced with faced with a learning task. Suppose, for example, that 100 variables encode five response variables and 95 regressors of varying inference importance. One could first employ a direct Fourier-Walsh diagonalization estimator and optimize only elements of $\mathbf{b}$ in $S^{(95)}_1$, $S^{(95)}_2$, and $S^{(95)}_3$. From these optimized quantities, one could find the set of 20 binary regressor variables most correlated with the response variables and then apply less regularized and restricted estimator variants using only these binary inputs. Such "variable search" approaches made possible by the presented estimator can make it powerful in the realm of machine learning over massive binary spaces. Methods to select hypercubes over which to estimate practical densities could be the focus of exciting future research.

\bibliographystyle{unsrt}

\end{document}